\documentclass[12pt]{article}
\usepackage[margin=1.25in]{geometry}
\usepackage{amsmath,amsthm,amssymb}
\usepackage{times}
\usepackage{enumerate}
\usepackage{setspace}
\usepackage[colorlinks,linkcolor=blue,anchorcolor=blue,citecolor=blue]{hyperref}
\usepackage{comment}
\usepackage{dsfont}
\usepackage{float}
\usepackage[title]{appendix}
\usepackage{natbib}
\usepackage{graphicx}
 
\theoremstyle{plain}

\newtheorem{theorem}{Theorem}

\newtheorem{lemma}{Lemma}

\newcommand{\R}{\mathbb{R}}

\makeatletter

\newcommand{\Rmnum}[1]{\expandafter\@slowromancap\romannumeral #1@}
 
\specialcomment{note}{\begingroup\color{blue} Note: }{\endgroup}
 
\title{Ambiguous Persuasion with Prior Ambiguity}
 
\author{Xiaoyu Cheng\footnote{Department of Economics, Florida State University, Tallahassee, FL, USA. E-mail: \href{mailto:xcheng@fsu.edu}{xcheng@fsu.edu}.}}
 
\begin{document}
\maketitle

\begin{abstract}
\cite{cheng2025ambiguous} establishes that in a persuasion game where both the sender and the receiver have Maxmin Expected Utility (MEU) preferences, the sender never strictly benefits from using ambiguous communication strategies over standard (non-ambiguous) ones. This note extends the analysis to environments with prior ambiguity, i.e., pre-existing ambiguity about the payoff-relevant state, and shows that, in the binary state and binary action case, the same no-gain result continues to hold.
\end{abstract}
 
\section{Introduction}
An ambiguous persuasion game generalizes the Bayesian persuasion model of \cite{kamenica2011bayesian} by allowing the sender to choose an \emph{ambiguous experiment} rather than a standard statistical experiment. \cite{cheng2025ambiguous} introduces an ex-ante formulation of the ambiguous persuasion game, in which the receiver, after observing the sender’s choice of experiment but before any message is realized, selects a message-contingent action plan. Under this formulation, \cite{cheng2025ambiguous} shows that when both players have Maxmin Expected Utility (MEU) preferences, the sender never strictly benefits from choosing ambiguous experiments. Following this approach, \cite{cheng2025persuasion} further characterizes when and how the sender can gain from ambiguous communication when both players instead have the less extreme smooth ambiguity preferences of \cite{klibanoff2005smooth}.

In both papers, to isolate the role of ambiguous communication, the main analysis focuses on the case where the players share a common \emph{single} prior belief about the payoff-relevant state, i.e., there is no pre-existing prior ambiguity. In this note, I extend the analysis of \cite{cheng2025ambiguous} to environments with prior ambiguity, where both players have MEU preferences and share a common ambiguous belief represented by a set of priors.\footnote{An extension of the analysis in \cite{cheng2025persuasion} to such environments is discussed in that paper.} Although the original argument cannot be applied directly, this note shows that, in the special case of binary state and binary action, the same no-gain result continues to hold. A full analysis of more general settings is left for future work.

\section{The Persuasion Game with Prior Ambiguity}
 
I adopt the notation and terminology from \cite{cheng2025persuasion} to formalize the persuasion game with prior ambiguity. Let $\Omega$ be a finite set of payoff-relevant states. The sender and receiver share a common ambiguous belief over $\Omega$, represented by a closed and convex set $P \subseteq \Delta(\Omega)$. Both players are assumed to have Maxmin Expected Utility (MEU) preferences. The receiver has a finite action set $A$. If the receiver chooses $a \in A$, the payoff to the sender is $u_s(a, \omega)$ and to the receiver is $u_r(a, \omega)$, when the state is $\omega \in \Omega$. A statistical experiment $\sigma$ is a mapping from $\Omega$ to $\Delta(M)$, where $M$ is a finite set of messages. 

The receiver's strategy is a mapping $\tau: M \rightarrow \Delta(A)$. Let $u_{i}(p, \sigma, \tau)$ denote the expected payoff of player $i \in \{s, r\}$ when the sender uses experiment $\sigma$, the prior is $p \in P$, and the receiver's strategy is $\tau$:
\begin{equation*}
u_{i}(p, \sigma, \tau) = \sum_{\omega, m, a} p(\omega)\sigma(m|\omega)\tau(a|m)u_i(a, \omega).
\end{equation*}
As the set of prior is fixed, further write
\begin{equation*}
  u_{i}(\sigma, \tau) = \min\limits_{p \in P}\sum_{\omega, m, a} p(\omega)\sigma(m|\omega)\tau(a|m)u_i(a, \omega).  
\end{equation*}

When the sender is constrained to use only statistical experiments, the sender's program is given by
\begin{equation}\label{program:statistical}
\begin{aligned}
\max_{\sigma} \quad & u_s(\sigma, \tau) \\
\text{s.t.} \quad & \tau \in BR(\sigma) = \mathop{\arg\max}_{\tau} u_r(\sigma, \tau), 
\end{aligned}
\end{equation}
where the receiver's program is more explicitly, 
\begin{equation*}
    \max_{\tau} \min\limits_{p \in P} \sum_{\omega, m, a} p(\omega)\sigma(m|\omega)\tau(a|m)u_r(a, \omega).
\end{equation*}


Next, the sender may choose an ambiguous experiment $\Sigma$, which is a closed and convex set of statistical experiments. Given an ambiguous experiment $\Sigma$ and the receiver's strategy $\tau$, payoff of player $i$ is given by 
\begin{equation*}
    U_{i}(\Sigma, \tau) = \min_{p \in P} \min_{\sigma \in \Sigma} \sum_{\omega, m, a} p(\omega)\sigma(m|\omega)\tau(a|m)u_i(a, \omega) = \min_{\sigma \in \Sigma} u_i(\sigma, \tau).
\end{equation*}
In this case, the sender's program becomes
\begin{equation}\label{program:ambiguous}
\begin{aligned}
\max_{\Sigma} \quad &  U_{s}(\Sigma, \tau) \\
\text{s.t.} \quad & \tau \in BR(\Sigma) = \mathop{\arg\max}_{\tau} U_r(\Sigma, \tau).
\end{aligned}
\end{equation}
And the receiver's program now becomes
\begin{equation*}
    \max_{\tau} \min\limits_{p \in P} \min\limits_{\sigma \in \Sigma} \sum_{\omega, m, a} p(\omega)\sigma(m|\omega)\tau(a|m)u_r(a, \omega) = \max_{\tau} U_{r}(\Sigma, \tau).
\end{equation*}

Say that ambiguous communication \emph{benefits} the sender if their optimal value in program \eqref{program:ambiguous} is strictly higher than that in program \eqref{program:statistical}. 

\section{Preliminary Analysis}
In this section, I present a revelation principle that simplifies the sender's program. In particular, it leads to a simple necessary condition for ambiguous communication to benefit the sender. 

\subsection{A Revelation Principle}
Say that an experiment $\sigma$ is \emph{canonical} if it maps $\Omega$ to $\Delta(A)$, i.e., its messages are action recommendations. Let $\tau^{*}$ denote the obedient strategy of the receiver, i.e., $\tau^{*}(a|a) = 1$ for all $a \in A$. I call such an experiment as itself \emph{obedient}. The next Lemma shows that it is without loss to let the sender choose among canonical and obedient experiments.

\begin{lemma}\label{lem:revelation_principle}
    For any $(\Sigma, \tau)$ such that $\tau \in BR(\Sigma)$, there exists a canonical and obedient ambiguous experiment $\Sigma^{*}$ such that $u_{i}(\sigma, \tau) = u_{i}(\sigma^{*}, \tau^{*})$ for all $i \in \{s, r\}$ and $\sigma \in \Sigma$. 
\end{lemma}

\begin{proof}[Proof of Lemma \ref{lem:revelation_principle}]
    Fix $(\Sigma, \tau)$ such that $\tau \in BR(\Sigma)$. For each $\sigma \in \Sigma$, define a canonical experiment $\sigma^{*}$ as follows: 
    \begin{equation*}
        \sigma^{*}(a|\omega) = \sum_{m} \sigma(m|\omega)\tau(a|m), \quad \forall a \in A, \omega \in \Omega.
    \end{equation*}
    Let $\Sigma^{*} = \{\sigma^{*} \mid \sigma \in \Sigma\}$, which is a closed and convex set of canonical experiments. Then for all $i \in \{s, r\}$, it holds that
    \begin{align*}
        u_{i}(\sigma, \tau) &= \min\limits_{p \in P} \sum_{\omega, m, a} p(\omega)\sigma(m|\omega)\tau(a|m)u_i(a, \omega) \\
        &= \min\limits_{p \in P} \sum_{\omega, a} p(\omega)\sigma^{*}(a|\omega)u_i(a, \omega) \\
        &= u_{i}(\sigma^{*}, \tau^{*}).
    \end{align*}
    
    Next, show that $\tau^{*} \in BR(\Sigma^{*})$. Suppose not, then there exists $\delta: A \rightarrow \Delta(A)$ such that
    \begin{equation*}
        U_{r}(\Sigma^{*}, \delta) > U_{r}(\Sigma^{*}, \tau^{*}).
    \end{equation*}
    Let $\tau'(a|m) = \sum_{a'} \delta(a|a')\tau(a'|m)$, then 
    \begin{align*}
        U_{r}(\Sigma, \tau') = U_{r}(\Sigma^{*}, \delta) > U_{r}(\Sigma^{*}, \tau^{*}) = U_{r}(\Sigma, \tau),
    \end{align*}
    thus, a contradiction. 
\end{proof}

Notice the same revelation principle holds for the case of statistical experiments as well. From this point on, all experiments are assumed to be canonical. The sender's program \eqref{program:statistical} can be rewritten as
\begin{equation}\label{program:statistical_revealed}
\begin{aligned}
\max_{\sigma} \quad & u_s(\sigma, \tau^{*})\\
\text{s.t.} \quad & \tau^{*} \in BR(\sigma).
\end{aligned}
\end{equation}  
Similarly, the sender's program \eqref{program:ambiguous} can be rewritten as
\begin{equation}\label{program:ambiguous_revealed}
\begin{aligned}
\max_{\Sigma} \quad &  U_{s}(\Sigma, \tau^{*})\\
\text{s.t.} \quad & \tau^{*} \in BR(\Sigma).
\end{aligned}
\end{equation}

Thus, equivalently, ambiguous communication benefits the sender if value of the program \eqref{program:ambiguous_revealed} is strictly greater than that of \eqref{program:statistical_revealed}. 

\subsection{A Necessary Condition}

Let $\sigma^{*}$ denote the optimal solution to the sender's program \eqref{program:statistical_revealed} and suppose that an ambiguous experiment $\Sigma$ is obedient and benefits the sender. Then it holds that 
\begin{align*}
    U_{s}(\Sigma, \tau^{*}) &= \min\limits_{p \in P} \min\limits_{\sigma \in \Sigma} u_s(\sigma, \tau^{*}) > \min\limits_{p \in P} u_s(\sigma^{*}, \tau^{*}),
\end{align*}
which further implies that for all $\sigma \in \Sigma$,
\begin{equation*}
    \min\limits_{p \in P} u_s(\sigma, \tau^{*}) > \min\limits_{p \in P} u_s(\sigma^{*}, \tau^{*}).
\end{equation*}
Because $\sigma^{*}$ is the optimal statistical experiment, it must be the case that for all $\sigma \in \Sigma$, $\tau^{*} \not\in BR(\sigma)$. In other words, if an ambiguous experiment benefits the sender, then it must not include any statistical experiment for which obedience is a best response. This observation leads to the following necessary condition. 

\begin{lemma}\label{lem:necessary_condition}
    If ambiguous communication benefits the sender, then there exists an ambiguous experiment $\Sigma$ such that $\tau^{*} \in BR(\Sigma)$ but for all $\sigma \in \Sigma$, $\tau^{*} \not\in BR(\sigma)$.
\end{lemma}

Given Lemma \ref{lem:necessary_condition}, it is thus essential to analyze the condition for obedience to hold under both statistical and ambiguous experiments. 

\section{Characterization of Obedience}

Let $\pi^{(p,\sigma)} \in \Delta(\Omega \times A)$ denote the joint distribution of states and actions induced by prior $p$ and experiment $\sigma$, i.e.,
\begin{equation*}
    \pi^{(p,\sigma)}(\omega, a) = p(\omega) \sigma(a|\omega).
\end{equation*}
Let $\Pi^{*}$ denote the set of joint distributions under which obedience holds, that is, 
\begin{equation*}
    \tau^{*} \in \mathop{\arg\max}_{\tau} \sum_{a, \omega} \pi(\omega, a) u_r(a', \omega)\tau(a'|a).
\end{equation*}
Equivalently, by linearity of the objective function, it holds that $\pi \in \Pi^{*}$ if and only if for every $a, b \in A$,
\begin{equation*}
    \sum_{\omega} \pi(\omega, a) [u_r(b, \omega) - u_r(a, \omega)] \leq 0.
\end{equation*}
As a result, $\Pi^{*}$ is the intersection of a finite number of half-spaces, thus it is a closed and convex set.

For each $\pi \in \Pi$, let $\pi_{a} \in \R^{|\Omega|}$ denote the vector with $\pi_{a}(\omega) = \pi(\omega, a)$ for all $\omega \in \Omega$. For $a, b \in A$, let $v_{a \rightarrow b} \in \R^{|\Omega|}$ denote the vector of the receiver's payoff difference when taking action $b$ instead of $a$ in each state, i.e., $v_{a \rightarrow b}(\omega) = u_r(b, \omega) - u_r(a, \omega)$ for all $\omega \in \Omega$. Therefore, $\pi \in \Pi^{*}$ if and only if 
\begin{equation}\label{eq:obedience_joint}
    \langle \pi_{a}, v_{a \rightarrow b} \rangle \leq 0, \quad \forall a, b \in A.
\end{equation}

\paragraph{Obedience under Statistical Experiments} Given a statistical experiment $\sigma$, the minimax theorem \citep{sion1958general} implies that $\tau^{*} \in BR(\sigma)$ if and only if it is part of a saddle point, i.e., there exists $p \in P$ such that 
\begin{equation*}
    p \times \sigma \in \Pi^{*}, \quad \text{and}\quad  p \in \mathop{\arg\min}_{p \in P} u_{r}(p, \sigma, \tau^{*}).
\end{equation*}
Let $P_{r}^{*}(\sigma)$ denote the set of worst-case priors when the receiver's strategy is $\tau^{*}$. Equivalently, $\tau^{*} \in BR(\sigma)$ if and only if 
\begin{equation}\label{eq:obedience_statistical}
    \exists \ p \in P_{r}^{*}(\sigma) \text{ such that } \langle (p \times \sigma)_{a}, v_{a \rightarrow b} \rangle \leq 0, \quad \forall a, b \in A.
\end{equation}

\paragraph{Obedience under Ambiguous Experiments} Given an ambiguous experiment $\Sigma$, let $\Pi(\Sigma)$ denote the set of joint distributions induced by $\Sigma$, that is,
\begin{equation*}
    \Pi(\Sigma) = \{\pi^{(p, \sigma)} \mid p \in P, \sigma \in \Sigma\}.
\end{equation*}
Notice that while $P$ and $\Sigma$ are both closed and convex sets, $\Pi(\Sigma)$ is not necessarily convex. Further let $\overline{co}(\Pi(\Sigma))$ denote its closed convex hull and define
\begin{equation*}
    K^{*}(\Sigma) := \mathop{\arg\min}_{\pi \in \overline{co}(\Pi(\Sigma))} \sum_{a, \omega} \pi(\omega, a) u_r(a, \omega),
\end{equation*}
i.e., the set of joint distributions in $\overline{co}(\Pi(\Sigma))$ that minimizes the receiver's payoff under $\tau^{*}$. Notice that $K^{*}(\Sigma)$ is an exposed face of the convex set $\overline{co}(\Pi(\Sigma))$, thus its extreme points are also extreme points of $\overline{co}(\Pi(\Sigma))$, taking the form of $\pi^{(p, \sigma)}$ for some $\sigma \in \Sigma$ and $p \in P_{r}^{*}(\sigma)$. 

Again, with the minimax theorem applied to the convex hull, it holds that $\tau^{*} \in BR(\Sigma)$ if and only if 
\begin{equation}\label{eq:obedience_ambiguous}
    \exists \ \pi \in K^{*}(\Sigma) \text{ such that } \langle \pi_{a}, v_{a \rightarrow b} \rangle \leq 0, \quad \forall a, b \in A.
\end{equation}

\bigskip 

When $P$ is a singleton as in \cite{cheng2025ambiguous}, obedience of $\Sigma$ implies the existence of $\sigma \in \Sigma$ such that $\tau^{*} \in BR(\sigma)$ under the singleton prior, which further establishes that ambiguous communication does not benefit the sender. The argument, however, does not extend to the case where $P$ is not a singleton: in that setting, not every point in $K^{*}(\Sigma)$ can be represented as a decomposition into $(p,\sigma)$ pairs. This could happen when $\pi \in K^{*}(\Sigma)$ is a convex combination of extreme points that are themselves induced by different priors. Therefore, it leaves open the possibility that obedience of $\Sigma$ does not imply obedience of any $\sigma \in \Sigma$ and that ambiguous communication may benefit the sender.

\section{No Gain with Binary State and Action}
This section focuses on the special case of binary state and action, i.e., $\Omega = \{\omega_{1}, \omega_{2}\}$ and $A = \{a, b\}$. By exploiting the simple geometry in this case, the following no gain result can be established.

\begin{theorem}\label{thm:binary_no_gain}
    Suppose $|\Omega| = |A| = 2$. Then, ambiguous communication does not benefit the sender.
\end{theorem}

\begin{proof}[Proof of Theorem \ref{thm:binary_no_gain}]

By Lemma \ref{lem:necessary_condition}, it suffices to show that for any $\Sigma$ such that $\tau^{*} \in BR(\Sigma)$, there exists $\sigma \in \Sigma$ such that $\tau^{*} \in BR(\sigma)$.

Let $P = [p_{L}, p_{U}] \subseteq [0,1]$ denote the set of priors, where I use $p$ to denote the probability of state $\omega_{1}$. Let $v = [v_{1}, v_{2}]^{\top} = v_{a \rightarrow b}$, i.e., $v_{i} = u_r(b, \omega_{i}) - u_r(a, \omega_{i})$ for $i = 1, 2$. If $v \geq 0$ or $v \leq 0$, then one action weakly dominates the other, and thus $\tau^{*}$ is a best response to any experiment (statistical or ambiguous) if and only if the action is recommended in all states. The conclusion thus holds trivially.

Applying a normalization to let $v_{1} = 1$ and $v_{2} = -k$ for some $k > 0$. Then obedience of $\pi$ is equivalent to
\begin{align*}
    & \langle \pi_{a}, v \rangle = \pi(\omega_{1}, a) - k \pi(\omega_{2}, a) \leq 0, \text{ and } \\
    & \langle \pi_{b}, -v \rangle = - \pi(\omega_{1}, b) + k \pi(\omega_{2}, b) \leq 0.
\end{align*}
Similarly, the obedience of $\sigma$ under prior $p$ is equivalent to
\begin{align*}
    &\Phi_{a \rightarrow b}(p, \sigma) := \langle (p \times \sigma)_{a}, v \rangle = p \sigma(a|\omega_{1}) - k(1-p)\sigma(a|\omega_{2}) \leq 0, \text{ and } \\
    &\Phi_{b \rightarrow a}(p, \sigma) := \langle (p \times \sigma)_{b}, -v \rangle = \Phi_{a \rightarrow b}(p, \sigma) + k - (1+k) p \leq 0.
\end{align*}
In this case, obedience of $\sigma$ is further equivalent to
\begin{equation*}
    \Phi_{a \rightarrow b}(p, \sigma) \leq \min \{0, (1+k)p - k\}.
\end{equation*}

Next, let $w = [w_{1}, w_{2}]$ denote the vector with $w_{i} = u_{r}(a, \omega_{i})$ for $i = 1, 2$. Then the receiver's expected payoff is given by
\begin{align*}
    u_{r}(p, \sigma, \tau^{*}) & = p \sigma(a|\omega_{1}) w_{1} + p \sigma(b|\omega_{1}) (w_{1} + 1) + (1-p) \sigma(a|\omega_{2}) w_{2} + (1-p) \sigma(b|\omega_{2}) (w_{2} - k) \\
    & = [w_{1} - w_{2} + (1-\sigma(a|\omega_{1})) + k(1-\sigma(a|\omega_{2}))] p + w_{2} - k(1-\sigma(a|\omega_{2})) := M(\sigma) p + N(\sigma).
\end{align*}
Thus the receiver's worst-case prior is given by
\begin{equation*}
    P_{r}^{*}(\sigma) = \begin{cases}
       \{ p_{L} \}, & \text{ if } M(\sigma) > 0, \\
       \{ p_{U} \}, & \text{ if } M(\sigma)  < 0, \\
        [p_{L}, p_{U}], & \text{ if } M(\sigma)  = 0.
    \end{cases}
\end{equation*}
Recall $K^{*}(\Sigma)$ is the set of joint distributions in $\overline{co}(\Pi(\Sigma))$ that minimizes the receiver's payoff under $\tau^{*}$. Define
\begin{equation*}
    S_{L}(\Sigma) = \{\sigma \mid \pi^{(p_{L}, \sigma)} \in K^{*}(\Sigma) \}, \quad S_{U}(\Sigma) = \{\sigma \mid \pi^{(p_{U}, \sigma)} \in K^{*}(\Sigma) \}.
\end{equation*}
By convexity of $\Sigma$, it follows that $S_{L}(\Sigma)$ and $S_{U}(\Sigma)$ are both convex sets.



When $\Sigma$ is obedient, there exists $\pi \in K^{*}(\Sigma)$ such that obedience of $\pi$ holds. As $\pi \in K^{*}(\Sigma)$, it can be represented as a convex combination of extreme points in $K^{*}(\Sigma)$. Notice that all extreme points of $K^{*}(\Sigma)$ take the form of either $\pi^{(p_{L}, \sigma)}$ for $\sigma \in S_{L}(\Sigma)$ or $\pi^{(p_{U}, \sigma)}$ for $\sigma \in S_{U}(\Sigma)$. Thus, by taking convex combinations within $S_{L}(\Sigma)$ and $S_{U}(\Sigma)$, there exist $\sigma_{L} \in S_{L}(\Sigma)$ and $\sigma_{U} \in S_{U}(\Sigma)$, and $\alpha \in [0,1]$ such that
\begin{equation*}
    \pi = \alpha \pi^{(p_{L}, \sigma_{L})} + (1-\alpha) \pi^{(p_{U}, \sigma_{U})}.
\end{equation*}
The obedience of $\pi$ thus implies that
\begin{align*}
    &\alpha \Phi_{a \rightarrow b}(p_{L}, \sigma_{L}) + (1-\alpha) \Phi_{a \rightarrow b}(p_{U}, \sigma_{U}) \leq 0, \text{ and } \\
    &\alpha \Phi_{a \rightarrow b}(p_{L}, \sigma_{L}) + (1-\alpha) \Phi_{a \rightarrow b}(p_{U}, \sigma_{U}) + k - (1+k)(\alpha p_{L} + (1-\alpha)p_{U})\leq 0. 
\end{align*}
Let $p_{\alpha} = \alpha p_{L} + (1-\alpha)p_{U}$. Then obedience of $\pi$ is equivalent to 
\begin{align*}
    &\alpha \Phi_{a \rightarrow b}(p_{L}, \sigma_{L}) + (1-\alpha) \Phi_{a \rightarrow b}(p_{U}, \sigma_{U}) \leq \min \{0, (1+k)p_{\alpha} - k\}.
\end{align*}
Since $M(\sigma_{L}) \geq 0$ and $M(\sigma_{U}) \leq 0$, there exists $\lambda \in [0,1]$ such that $M(\hat{\sigma}) = 0$ for $\hat{\sigma} = \lambda \sigma_{L} + (1-\lambda) \sigma_{U}$. In this case, $P_{r}^{*}(\hat{\sigma}) = [p_{L}, p_{U}]$. I next show that $\hat{\sigma}$ is obedient under $p_{\alpha}$. 

First notice that 
\begin{align*}
    u_{r}(p, \sigma, \tau^{*}) = p (w_{1} + 1) + (1-p) (w_{2} - k) - \Phi_{a \rightarrow b}(p, \sigma),
\end{align*}
which further implies that 
\begin{equation*}
    u_{r}(p, \sigma, \tau^{*}) \geq u_{r}(p, \sigma', \tau^{*}) \iff \Phi_{a \rightarrow b}(p, \sigma) \leq \Phi_{a \rightarrow b}(p, \sigma').
\end{equation*}
By definition of $S_{L}(\Sigma)$ and $S_{U}(\Sigma)$, it then holds that 
\begin{equation*}
    \Phi_{a \rightarrow b}(p_{L}, \sigma_{U}) \leq \Phi_{a \rightarrow b}(p_{L}, \sigma_{L}), \text{ and } \Phi_{a \rightarrow b}(p_{U}, \sigma_{L}) \leq \Phi_{a \rightarrow b}(p_{U}, \sigma_{U}), \quad \forall \sigma \in \Sigma.
\end{equation*}
Then by linearity of $\Phi_{a \rightarrow b}(p, \cdot)$ in $\sigma$, it follows that 
\begin{equation*}
    \Phi_{a \rightarrow b}(p_{L}, \hat{\sigma}) \leq \Phi_{a \rightarrow b}(p_{L}, \sigma_{L}), \text{ and } \Phi_{a \rightarrow b}(p_{U}, \hat{\sigma}) \leq \Phi_{a \rightarrow b}(p_{U}, \sigma_{U}).
\end{equation*}
Furthermore, by linearity of $\Phi_{a \rightarrow b}(\cdot, \hat{\sigma})$ in $p$, it holds that
\begin{align*}
    \Phi_{a \rightarrow b}(p_{\alpha}, \hat{\sigma}) &= \alpha \Phi_{a \rightarrow b}(p_{L}, \hat{\sigma}) + (1-\alpha) \Phi_{a \rightarrow b}(p_{U}, \hat{\sigma}) \\
    &\leq \alpha \Phi_{a \rightarrow b}(p_{L}, \sigma_{L}) + (1-\alpha) \Phi_{a \rightarrow b}(p_{U}, \sigma_{U}) \leq \min \{0, (1+k)p_{\alpha} - k\}.
\end{align*}
Thus, obedience condition \eqref{eq:obedience_statistical} holds for $\hat{\sigma}$ at $p_{\alpha}$. Since $P_{r}^{*}(\hat{\sigma}) = [p_{L}, p_{U}]$, it follows that $\tau^{*} \in BR(\hat{\sigma})$.
\end{proof}

\bibliographystyle{econ}
\bibliography{references.bib}
 
\end{document}